\documentclass[11pt]{article}
 \usepackage{amsfonts}
 \usepackage{mathrsfs}
 \usepackage{bbm}
 \usepackage{amsthm}
\usepackage{amsmath,amssymb}
\newtheorem{theorem}{Theorem}[section]
\newtheorem{lemma}[theorem]{Lemma}

\theoremstyle{definition}

\theoremstyle{remark}
\newtheorem{remark}[theorem]{Remark}

\theoremstyle{algorithm}
\newtheorem{algorithm}[theorem]{Algorithm}

\theoremstyle{corollary}

\theoremstyle{example}
\newtheorem{example}[theorem]{Example}

\usepackage{amssymb}
\usepackage{epsfig}
\usepackage{times}
\title{Distributed Detection Fusion via Monte Carlo Importance Sampling}
\author{Hang Rao,  Xiaojing Shen\thanks{This work was supported in part by  the open research funds of BACC-STAFDL of China
under Grant No. 2015afdl010,  the special funds of NEDD of China under Grant No. 201314, the NSF No. 61273074£© and the PCSIRT1273. Hang Rao, Xiaojing Shen (corresponding author), Yunmin Zhu and Jianxin Pan are  with Department of Mathematics, Sichuan
University, Chengdu, Sichuan 610064, China. E-mail: shenxj@scu.edu.cn,  ymzhu@scu.edu.cn, jianxin.pan@manchester.ac.uk.}, ~Yunmin Zhu and Jianxin Pan}

\begin{document}
 \maketitle
\begin{abstract}
Distributed detection fusion with high-dimension conditionally dependent observations is known to be a challenging problem. When a fusion rule is fixed, this paper attempts to make progress on this problem for the large sensor networks by proposing a new Monte Carlo framework. Through the Monte Carlo importance sampling, we derive a necessary condition for optimal sensor decision rules in the sense of minimizing the approximated Bayesian cost function. Then, a Gauss-Seidel/person-by-person optimization algorithm can be obtained to search the optimal sensor decision rules. It is proved that the discretized algorithm is finitely convergent. The complexity of the new algorithm is  $O(LN)$ compared with  $O(LN^L)$ of the previous algorithm where $L$ is the number of sensors and $N$ is a constant. Thus, the proposed methods allows us to design the large sensor networks with general high-dimension dependent observations. Furthermore, an interesting result is that, for the fixed AND or OR fusion rules, we can analytically derive the optimal solution in the sense of minimizing the approximated Bayesian cost function. In general, the solution of the Gauss-Seidel algorithm is only local optimal. However, in the new framework, we can prove that the solution of Gauss-Seidel algorithm is same as the analytically optimal solution in the case of the AND or OR fusion rule. The typical examples with dependent observations and large number of sensors are examined under this new framework. The results of numerical examples demonstrate the effectiveness of the new algorithm.
\end{abstract}

\noindent{\bf keywords:} Distributed detection, Monte Carlo importance sampling, dependent observations, sensor decision rule, fusion rule

\section{Introduction}\label{sec_1}

Distributed signal detection has received significant attention in  surveillance %military and  civilian
applications over the past thirty years \cite{Tenney-Sandell81,Tsitsiklis-Athans85,Chair-Varshney86, Varshney97,Blum-Kassam-Poor97,Chen-Willett05,Chen-Tong-Varshney06,Zhu-Zhou-Shen-Song-Luo12}. %% Tsitsiklis93
Tenney and Sandell \cite{Tenney-Sandell81} firstly considered Bayesian formulation of distributed detection for parallel sensor network structures and proved that the optimal decision rules at the sensors are likelihood ratio (LR) for conditionally independent sensor observations. However, the optimal thresholds of LR at individual sensors can be only obtained by
solving a set of coupled nonlinear equations. When the sensor decision rules are fixed, Chair and Varshney \cite{Chair-Varshney86} derived an optimal fusion rule based on the LR test.
%\textcolor[rgb]{0.98,0.00,0.00}{Moreover, the case of correlated local decisions has also been studied} \cite{Drakopoulos-Lee91,Kam-Zhu-Gray-Steven92}.%some progresses have been made for correlated local decisions cases \cite{Drakopoulos-Lee91,Kam-Zhu-Gray-Steven92}. %%Drakopoulos-Lee91,Kam-Zhu-Gray-Steven92
For conditionally independent sensor observations, many excellent results on distributed detection have been derived and are summarized  in \cite{Varshney97} and references therein. The emerging wireless sensor networks \cite{Chen-Tong-Varshney06} motivated the optimality of LR thresholds to be extended to non-ideal detection systems in which sensor outputs are to be communicated through noisy, possibly coupled channels to the fusion center \cite{Chen-Willett05,Kashyap06,Chen-Chen-Varshney09}. %%Kashyap06

There is much less attention on the studies of sensor decision rules for generally dependent observations which were considered to be difficult (see, e.g., \cite{Tenney-Sandell81,Tsitsiklis-Athans85,Hoballah-Varshney89}). Tsitsiklis and Athans \cite{Tsitsiklis-Athans85} provided a rigorous mathematical analysis to demonstrate the computational difficulty in obtaining the optimal sensor decision rules for dependent sensor observations. However, some progresses have been made for the special dependent observations cases (see, e.g., \cite{Drakopoulos-Lee91,Kam-Zhu-Gray-Steven92,Blum-Kassam92,Tang-Pattipati-Kleinman92,Chen-Papamarcou95,Yan_Blum01, Willett-Swaszek-Blum00}). %%Yan01
Willett et al. \cite{Willett-Swaszek-Blum00} discussed difficulties for dealing with dependent observations. Zhu et al.\cite{Zhu-Blum-Luo-Wong00} proposed a computationally efficient iterative algorithm which computes a discrete approximation of the optimal sensor decision rules for general dependent observations and a fixed fusion rule. This algorithm converges in finite steps. In \cite{Shen-Zhu-He-You11}, the authors developed an efficient algorithm to simultaneously search for the optimal fusion rule and the optimal sensor rules by combining the methods of Chair and Varshney \cite{Chair-Varshney86} and Zhu et al. \cite{Zhu-Blum-Luo-Wong00}. Recently, a new framework for distributed detection with conditionally dependent observations was introduced in \cite{Chen-Chen-Varshney12}, which can identify several classes of problems with dependent observations whose optimal sensor decision rules resemble the ones for the independent case.

Although large sensor networks have attracted much attention in both theory and application \cite{Tsitsiklis88,Chamberland-Veeravalli03,Niu-Varshney-05}, the studies of sensor decision rules for large sensor networks with general dependent observations have had little progress. The fundamental reason is that the computation complexity is $O(LN^L)$ for the previous algorithms, where  $L$ is the number of sensors and $N$ is a given constant. In this paper, we propose a new Monte Carlo framework to overcome the limitation of the discretized algorithms in \cite{Zhu-Blum-Luo-Wong00,Shen-Zhu-He-You11} for the large sensor networks. Through the Monte Carlo importance sampling \cite{Liu01,Robert05}, the Bayesian cost function is approximated by the sample average by the strong law of large number. Then, we derive a necessary condition for optimal sensor decision rules so that a Gauss-Seidel optimization algorithm can be obtained to search the optimal sensor decision rules. It is proved that the new discretized algorithm is finitely convergent. The complexity of the new algorithm is order of $O(LN)$ compared with $O(LN^L)$ of the algorithms in \cite{Zhu-Blum-Luo-Wong00,Shen-Zhu-He-You11}.  Thus, the proposed methods allows us to design the large sensor networks with general dependent observations. Furthermore, an interesting result is that, for the fixed AND or OR fusion rules, we can analytically derive the optimal solution in the sense of minimizing the approximated Bayesian cost function. In general, the solution of the Gauss-Seidel algorithm is only local optimal. However, in the new framework, we can prove that the solution of Gauss-Seidel algorithm is same as the analytically optimal solution when the fusion rule is the AND or OR.
%Numerical examples show that the performance of the new algorithm with $O(2N)$ computations is close to that of the  algorithm in \cite{Zhu-Blum-Luo-Wong00} with  $O(2N^2)$ computations.
The typical examples with dependent observations and large number of sensors are examined under this new framework.  The results of numerical examples demonstrate the effectiveness of the new algorithm. The performance of the new algorithm based on Mixture-Gaussian trial distribution is better than that based on Gaussian trial distribution.
%In addition, as far as the Monte Carlo importance sampling is considered, the performance of the mixture Gaussian trial distribution is better than that Gaussian trial distribution. The reason may be that the former is closer to the optimal trial distribution.

The rest of the paper is organized as follows. Preliminaries are given in Section \ref{sec_2}, including problem formulation and  Monte Carlo approximation of the cost function. In Section \ref{sec_3}, necessary conditions for optimal sensor decision rules are given. In Section \ref{sec_4}, a Gauss-Seidel iterative algorithm is presented based on the necessary conditions. The convergence of this algorithm is proved. For the fixed AND or OR fusion rules, the optimal solution in the sense of minimizing the approximated Bayesian cost function can be analytically derived. Moreover, we prove that the solution of Gauss-Seidel algorithm is same as the analytically optimal solution in the case of the AND or OR fusion rule. In Section \ref{sec_5}, numerical examples are given that exhibit the effectiveness of the new algorithm class. In Section \ref{sec_6}, we draw conclusions.

\section{Preliminaries}\label{sec_2}
\subsection{Problem formulation}\label{sec_2_1}
The $L$-sensor Bayesian detection model with two hypotheses $H_0$ and $H_1$ are considered as follows. A parallel architecture is assumed.
%The observational data are $y_1$, $y_2$,\ldots, $y_L$.
The $i$th sensor compresses the $n_i$-dimensional vector observation $y_i$ to one bit:      $I_i(y_i):\mathbb{R}^{n_i}\rightarrow\{0,1\},i=1,\ldots,L$.
In this paper, we consider deterministic (non-randomized) decision rules. When the fusion rule $F$ is fixed , the distributed multisensor Bayesian decision problem is to minimize the following Bayesian cost function by optimizing the sensor decision rule $I_1(y_1),\ldots,I_L(y_L)$,
\begin{eqnarray}
\nonumber &&C(I_1(y_1),\ldots,I_L(y_L);F)\\
\nonumber &=&C_{00}P_{0}P(F=0|H_0)+C_{01}P_{1}P(F=0|H_1)\\
\label{Eq_1_1} &&+C_{10}P_{0}P(F=1|H_0)+C_{11}P_{1}P(F=1|H_1),
\end{eqnarray}
where $C_{i j}$ are the known cost coefficients, $P_0$ and $P_1$ are the prior probabilities for the hypotheses $H_0$ and $H_1$, and $P(F=i|H_j)$ is the probability that the fusion center decides for hypothesis $i$ given hypothesis $H_j$ is true.   The
general form of the binary fusion rule  $F$ is denoted by
an indicator function on a set $S=\{(u_1,\ldots,u_L):u_i=0/1, i=1,2,\ldots,L\}$: %$S=\{(I_1(y_1),\ldots,I_L(y_L))\}$:
\begin{eqnarray}
\label{Eq_1_2} F((u_1,\ldots,u_L)): S\rightarrow\{0,1\}.
\end{eqnarray}
Note that a fusion rule is a binary division of the set $S$ and the
number of elements of the set $S$ is $2^L$, thus there exists
$2^{2^L}$ fusion rules. Let $s_k$ be the $k$-th element of $S$,
$k=1,\ldots,2^L$. Every $s_k$ is $L$-dimensional vector and
$s_k(i)=0$ or $1, i=1,\ldots,L$. For convenience, we denote sets
$S_0$ and $S_1$ as the elements in $S$ for which the algorithm took
decision $H_0$ and $H_1$ respectively, i.e.
%we denote that the sets $S_0$
%and $S_1$ are the elements in $S$ decided to $H_0$ and $H_1$
%respectively,
\begin{eqnarray}
\label{Eq_1_3}
S_0&=&\{s_k:~F(s_k)=0,~k=1,\ldots,2^L\},\\
\label{Eq_1_4}S_1&=&\{s_k:~F(s_k)=1,~k=1,\ldots,2^L\}.
\end{eqnarray}
Moreover, we let
$\Omega=\mathbb{R}^{n_1}\times\ldots\times\mathbb{R}^{n_L}$ and
denote
\begin{eqnarray}
\nonumber \Omega_0&=&\{(y_1,\ldots,y_L):~I_1(y_1)=s_k(1),\ldots,I_L(y_L)=s_k(L);\\
\label{Eq_1_5}&&~~~~~~~~~~~~~~~~~~~~F(s_k)=0,~k=1,\ldots,2^L\},\\~
\nonumber \Omega_1&=&\{(y_1,\ldots,y_L):~I_1(y_1)=s_k(1),\ldots,I_L(y_L)=s_k(L);\\
\label{Eq_1_6}&&~~~~~~~~~~~~~~~~~~~~F(s_k)=1,~k=1,\ldots,2^L\}.
\end{eqnarray}
Obviously, $S=S_0\bigcup S_1$ and  $\Omega=\Omega_0\bigcup
\Omega_1$. Suppose that $p(y_1,y_2,\ldots,y_L|H_1)$ and
$p(y_1,y_2,\ldots,y_L|H_0)$ are the known  conditional joint
probability density functions under each hypothesis.

Substituting the definitions of fusion rule $F$ and sensor decision rule $I_i(y_i)$ into (\ref{Eq_1_1}) and simplifying, we have
\begin{eqnarray}%
\nonumber&&C(I_1(y_1),\ldots,I_L(y_L);F)\\
\nonumber&=&C_{10}P_{0}+C_{11}P_{1}+\int_{\Omega_0}\{[P_1(C_{01}-C_{11})p(y_1,\ldots,y_L|H_1)]\nonumber\\
\nonumber&&-[P_0(C_{10}-C_{00})p(y_1,\ldots,y_L|H_0)]\}dy_1\cdots dy_L \nonumber \\
\label{Eq_1_7}
&=&c+\int_{\Omega}I_{\Omega_0}(y_1,\ldots,y_L)\hat{L}(y_1,\ldots,y_L)dy_1\ldots dy_L,
\end{eqnarray}
where $I_{\Omega_0}(y_1,\ldots,y_L)$ is an indicator function on $\Omega_0$,
\begin{eqnarray}%
\label{Eq_1_8}\hat{L}(y_1,\ldots,y_L)=ap(y_1,\ldots,y_L|H_1)-bp(y_1,\ldots,y_L|H_0),\\
\label{Eq_1_9}a=P_1(C_{01}-C_{11}), \qquad  b=P_0(C_{10}-C_{00}), \qquad c=C_{10}P_{0}+C_{11}P_{1}.
\end{eqnarray}
$a$, $b$, $c$ are fixed constants.

%Our goal is to select a set of optimal sensor decision rules $I_1(y_1), \ldots, I_L(y_L)$ for a  fixed fusion rule $F$ such that (\ref{Eq_2}) is minimum. To achieve the desired goal, we firstly transform

The indicator function $I_{\Omega_0}(y_1,\ldots,y_L)$ can be written as $L$ equivalent polynomials of the sensor decision rules $I_1(y_1), \ldots, I_L(y_L)$ and the fusion rule $F$ as follows (see \cite{Shen-Zhu-He-You11}):
%\begin{eqnarray}%
%\nonumber I_{\Omega_0}&=&I_{\{(y_1,\ldots,y_L):I_1(y_1)=s_k(1),\ldots,I_L(y_L)=s_k(L);F(s_k)=0,k=1,\ldots,2^L\}}\\
%\nonumber &=&\sum_{k=1}^{2^L}I_{\{(y_1,\ldots,y_L):I_1(y_1)=s_k(1),\ldots,I_L(y_L)=s_k(L);F(s_k)=0,k=1,\ldots,2^L\}}\\
%\nonumber &=&\sum_{k=1}^{2^L}\{I_{\{y_1:I_1(y_1)=s_k(1)\}} \cdot \cdots   I_{\{y_L:I_L(y_L)=s_k(L)\}}I_{\{s_k:F(s_k)=0\}}\} \\
%\nonumber &=&\sum_{k=1}^{2^L}\{[s_k(1)I_1(y_1)+(1-s_k(1))(1-I_1(y_1))]\cdot \cdots \\
%\label{Eq_1_00} &&\cdot [s_k(L)I_L(y_L)+(1-s_k(L))(1-I_L(y_L))][1-F(s_k)]\}
%\end{eqnarray}
%To show that the solution of the optimization problem satisfies a fixed point type necessary condition, we need to make sure some transformations of (\ref{Eq_5}) as follows:
%\begin{eqnarray}%
%I_{\Omega_0}&=&\sum_{k=1}^{2^L}\{[-s_k(1)(1-I_1(y_1))+s_k(1)+(1-s_k(1))(1-I_1(y_1))]\cdot\\
%\nonumber &&\cdots\cdot[s_k(L)I_L(y_L)+(1-s_k(L))(1-I_L(y_L))][1-F(s_k)]\} \\
%\nonumber &=&(1-I_1(y_1))\sum_{k=1}^{2^L}\{[(1-2s_k(1))][s_k(2)I_2(y_2)+(1-s_k(2))(1-I_2(y_2))]\cdot \\
%\nonumber &&\cdots\cdot[s_k(L)I_L(y_L)+(1-s_k(L))(1-I_L(y_L))][1-F(s_k)]\} \\
%\nonumber &&+\sum_{k=1}^{2^L}\{s_k(1)[s_k(2)I_2(y_2)+(1-s_k(2))(1-I_2(y_2))]\cdot \\
%\nonumber && \cdots\cdot[s_k(L)I_L(y_L)+(1-s_k(L))(1-I_L(y_L))][1-F(s_k)]\} \\
% \label{Eq_6} &\triangleq&[1-I_1(y_1)]P_{11}(I_2(y_2),\ldots,I_L(y_L);F) +P_{12}(I_2(y_2),\ldots,I_L(y_L);F),
%\end{eqnarray}
%where...
%
%Similarly, we have

\begin{eqnarray}%
%\nonumber I_{\Omega_0}(y_1,y_2,\ldots,y_L)&=&[1-I_1(y_1)]P_{11}(I_2(y_2),\ldots,I_L(y_L);F) \\
\nonumber I_{\Omega_0}(y_1,\ldots,y_L)%%&=&I_{\{(y_1,\ldots,y_L):I_1(y_1)=s_k(1),\ldots,I_L(y_L)=s_k(L);F(s_k)=0,k=1,\ldots,2^L\}}\\
%%\nonumber &=&\sum_{k=1}^{2^L}I_{\{(y_1,\ldots,y_L):I_1(y_1)=s_k(1),\ldots,I_L(y_L)=s_k(L);F(s_k)=0\}}\\
%%\nonumber &=&\sum_{k=1}^{2^L}\{I_{\{y_1:I_1(y_1)=s_k(1)\}} \cdot \cdots   I_{\{y_L:I_L(y_L)=s_k(L)\}}I_{\{s_k:F(s_k)=0\}}\} \\
%%\nonumber &=&\sum_{k=1}^{2^L}\{[s_k(1)I_1(y_1)+(1-s_k(1))(1-I_1(y_1))]\cdot \cdots \\
%&=&\sum_{k=1}^{2^L}\{[s_k(1)I_1(y_1)+(1-s_k(1))(1-I_1(y_1))]\cdot \cdots \\
%%\label{Eq_1_00}
%\nonumber &&\cdot [s_k(L)I_L(y_L)+(1-s_k(L))(1-I_L(y_L))][1-F(s_k)]\}\\
%%ɾµô¶àÓàµÄ¹«Ê½¼ÇºÅ
\nonumber &=&[1-I_1(y_1)]P_{11}(I_2(y_2),\ldots,I_L(y_L);F) \\
\label{Eq_1_10}&&+P_{12}(I_2(y_2),\ldots,I_L(y_L);F),\\
%\nonumber&=&[1-I_2(y_2)]P_{21}(I_1(y_1),I_3(y_3),\ldots,I_L(y_L);F)\\
%\nonumber&&+P_{22}(I_1(y_1),I_3(y_3),\ldots,I_L(y_L);F)\\
%\nonumber&=&[1-I_3(y_3)]P_{31}(I_1(y_1),I_2(y_2),I_4(y_4),\ldots,I_L(y_L);F)\\
%\nonumber&&+P_{32}(I_1(y_1),I_2(y_2),I_4(y_4),\ldots,I_L(y_L);F)\\
\nonumber&&\cdots\cdots\\
\nonumber&=&[1-I_L(y_L)]P_{L1}(I_1(y_1),I_2(y_2),\ldots,I_{L-1}(y_{L-1});F)\\
\label{Eq_1_11}&&+P_{L2}(I_1(y_1),I_2(y_2),\ldots,I_{L-1}(y_{L-1});F)
\end{eqnarray}
where, for $j=1,\ldots, L$,
\begin{eqnarray}%
\nonumber&&P_{j1}(I_1(y_1),\ldots,I_{j-1}(y_{j-1}),I_{j+1}(y_{j+1}),\ldots,I_{L}(y_{L});F)\\
\label{Eq_1_12}&\triangleq&\sum_{k=1}^{2^L}\{[1-F(s_k)][1-2s_k(j)]\prod_{m=1,m\neq j}^{2^L}[s_k(m)I_m(y_m)+(1-s_k(m))(1-I_m(y_m))]\} \\[5mm]
\nonumber&&P_{j2}(I_1(y_1),\ldots,I_{j-1}(y_{j-1}),I_{j+1}(y_{j+1}),\ldots,I_{L}(y_{L});F) \\
\label{Eq_1_13} &\triangleq&\sum_{k=1}^{2^L}\{[1-F(s_k)]s_k(j)\prod_{m=1,m\neq j}^{2^L}[s_k(m)I_m(y_m)+(1-s_k(m))(1-I_m(y_m))]\}
\end{eqnarray}
Note that both $P_{j1}(I_1(y_1),\ldots,I_{j-1}(y_{j-1}),I_{j+1}(y_{j+1}),\ldots,I_{L}(y_{L});F)$ and $P_{j2}(I_1(y_1),\ldots,I_{j-1}(y_{j-1}),\\
I_{j+1}(y_{j+1}),\ldots,I_{L}(y_{L});F)$ are  independent of $I_j(y_j)$ for $j=1,\ldots,L$. For convenience, we also denote them by $P_{j1}(\cdot)$, $P_{j2}(\cdot)$, respectively. Moreover, (\ref{Eq_1_12}) is also a key equation in the following results.

\subsection{Monte Carlo importance sampling}\label{sec_2_2}
In this section, we present an approximation of the cost function (\ref{Eq_1_7}) by Monte Carlo importance sampling (see, e.g.,  \cite{Liu01,Robert05}). More specifically, assume that the samples $Y_{1},\ldots,Y_{N}$ are from population $Y$ with a given trial distribution $g(y_1,y_2,\ldots,y_L)$, where $Y_i=[Y_{1i}, Y_{2i},\ldots,Y_{Li}]^T$.  From (\ref{Eq_1_7}),
\begin{eqnarray}%
\nonumber&&C(I_1(y_1),\ldots,I_L(y_L);F)\\
%%ɾµô¶àÓàµÄ¹«Ê½¼ÇºÅ
%\label{Eq_1_14}&=&\int_{\Omega}I_{\Omega_0}(y_1,y_2,\ldots,y_L)\hat{L}(y_1,y_2,\ldots,y_L)dy_1\ldots dy_L+c\\
\label{Eq_1_15}&=&\int_{\Omega}\frac{I_{\Omega_0}(y_1,y_2,\ldots,y_L)\hat{L}(y_1,y_2,\ldots,y_L)g(y_1,y_2,\ldots,y_L)}{g(y_1,y_2,\ldots,y_L)}dy_1\ldots dy_L+c \\
\label{Eq_1_16}&=&\mathbb{E}_g \frac{I_{\Omega_0}(Y)\hat{L}(Y)}{g(Y)}+c \\
\label{Eq_1_17}&\approx&\frac{1}{N}\sum_{i=1}^{N}\frac{I_{\Omega_0}(Y_{1i},Y_{2i},\ldots,Y_{Li})\hat{L}(Y_{1i},Y_{2i},\ldots,Y_{Li})}{g(Y_{1i},Y_{2i},\ldots,Y_{Li})}+c\\
\label{Eq_1_18}&\triangleq&C_{MC}(I_1(y_1),\ldots,I_L(y_L);F,N),
\end{eqnarray}
where $g(y_1,y_2,\ldots,y_L)$ is the trial density such that (\ref{Eq_1_15}) is well-defined. (\ref{Eq_1_16}) is from $Y\sim $ $g(y_1,y_2,$ $\ldots,y_L)$.  (\ref{Eq_1_17}) is denoted by $C_{MC}(I_1(y_1),\ldots,I_L(y_L);F,N)$. Based on the strong law of large number, (\ref{Eq_1_16}) can be approximated by (\ref{Eq_1_17}), i.e.,  $C_{MC}(I_1(y_1),\ldots,I_L(y_L);F,N)\rightarrow C(I_1(y_1),\ldots,$ $I_L(y_L);F),$ ~a.s. as $N\rightarrow\infty$. The optimal trial distribution is $g(y_1,y_2,\ldots$, $y_L)\varpropto |I_{\Omega_0}(y_1,y_2,\ldots,y_L)$ $\hat{L}(y_1,y_2,\ldots,y_L)|$ (see, e.g., \cite{Liu01,Robert05}). By (\ref{Eq_1_10}), (\ref{Eq_1_11}) and (\ref{Eq_1_17}), so that we have

\begin{eqnarray}%
\nonumber  && C_{MC}(I_1(y_1),\ldots,I_L(y_L);F,N)\\[3mm]
%\nonumber  &=&\frac{1}{N}\sum_{i=1}^{N}\{[1-I_1(Y_{1i})]P_{11}(I_2(Y_{2i}),\ldots,I_L(Y_{Li});F)\\
%\nonumber      &&+P_{12}(I_2(Y_{2i}),\ldots,I_L(Y_{Li});F)\}\frac{\hat{L}(Y_{1i},Y_{2i},\ldots,Y_{Li})}{g(Y_{1i},Y_{2i},\ldots,Y_{Li})}+c\\
\nonumber &=&\frac{1}{N}\sum_{i=1}^{N}\frac{[1-I_1(Y_{1i})]P_{11}(I_2(Y_{2i}),\ldots,I_L(Y_{Li});F)\hat{L}(Y_{1i},Y_{2i},\ldots,Y_{Li})}{g(Y_{1i},Y_{2i},\ldots,Y_{Li})}\\
\label{Eq_1_19}&&+\frac{1}{N}\sum_{i=1}^{N}\frac{P_{12}(I_2(Y_{2i}),\ldots,I_L(Y_{Li});F)\hat{L}(Y_{1i},Y_{2i},\ldots,Y_{Li})}{g(Y_{1i},Y_{2i},\ldots,Y_{Li})}+c\\
\nonumber     &&\cdots\cdots  \\
\nonumber      &=&\frac{1}{N}\sum_{i=1}^{N}\frac{[1-I_L(Y_{Li})]P_{L1}(I_1(Y_{1i}),\ldots,I_{L-1}(Y_{(L-1)i});F)\hat{L}(Y_{1i},Y_{2i},\ldots,Y_{Li})}{g(Y_{1i},Y_{2i},\ldots,Y_{Li})}\\
\label{Eq_1_20} &&+\frac{1}{N}\sum_{i=1}^{N}\frac{P_{L2}(I_1(Y_{1i}),\ldots,I_{L-1}(Y_{(L-1)i});F)\hat{L}(Y_{1i},Y_{2i},\ldots,Y_{Li})}{g(Y_{1i},Y_{2i},\ldots,Y_{Li})}+c
\end{eqnarray}

\section{Necessary Conditions For Optimum Sensor Decision Rules} \label{sec_3}
The distributed detection fusion problem is to minimize the Bayesian cost function $C(I_1(y_1),$ $\ldots,I_L(y_L);$ $F)$ (\ref{Eq_1_7}). Based on the Monte Carlo approximation  (\ref{Eq_1_18}), we concentrate on selecting a set of optimal sensor decision rules $I_1(y_1),\ldots,I_L(y_L)$ such that the approximated cost function $C_{MC}(I_1(y_1),\ldots,I_L(y_L);F,N)$ is minimum.

Firstly, we prove that the
minimum of the $C_{MC}(I_1(y_1),\ldots,I_L(y_L);F,N)$ cost functional converges to the infimum of the cost function
$C(I_1,\ldots,I_L;F)$ as the sample size $N$ tends to infinity, under
some mild assumptions. Since the deterministic (non-randomized) decision rules are considered in this paper, in the following sections, we assume that the samples drawn from the trial distribution have been fixed so that $C_{MC}(I_1(y_1),\ldots,I_L(y_L);F,N)$  has no randomness.

\begin{theorem}\label{thm_001}
Let $C_{inf}$ be the infimum of $C(I_1,\ldots,I_L;F)$ and $C_{MC}^*(I_1,\ldots,I_L;F,N)$ be the minimum of the Monte Carlo approximation
$C_{MC}(I_1,\ldots,I_L;F,N)$ (\ref{Eq_1_18}) where $I_1,\ldots,I_L$ are decision variables. If $C_{MC}(I_1,\ldots,I_L;F,N)$ satisfies
\begin{eqnarray}%
\label{Eq_3_0001}|C(I_1,\ldots,I_L;F)- C_{MC}(I_1,\ldots,I_L;F,N)|<\frac{\delta}{\sqrt{N}},
\end{eqnarray}%
where the constant $\delta$ does not depend on $I_1,\ldots,I_L, F$ and $N$,
then we have
\begin{eqnarray}
\label{Eq_3_001} \lim_{N\rightarrow\infty}C_{MC}^*(I_1,\ldots,I_L;F,N)=\inf_{I_{1},\cdots,I_{L}}C(I_1,\ldots,I_L;F) \triangleq C_{inf}.
\end{eqnarray}
\end{theorem}

\begin{proof}
%%ɾµô¶àÓàµÄ¹«Ê½¼ÇºÅ
By the definition of $C_{inf}$, for arbitrary $\epsilon > 0$, there exists a set of sensor rules $(I_1,\ldots,I_L)$ such that
\begin{eqnarray}
\nonumber   C(I_1,\ldots,I_L;F)\leq C_{inf}+\frac{1}{2}\epsilon
\end{eqnarray}
Since definition of  $C_{MC}(I_1,\ldots,I_L;F,N)$ and (\ref{Eq_3_0001}), there exists $N^{*}=(\frac{2\delta}{\epsilon})^2>0$ such that for any $N\geq N^{*}$
\begin{eqnarray}
\nonumber  C_{MC}(I_1,\ldots,I_L;F,N) \leq C(I_1,\ldots,I_L;F)+\frac{1}{2}\epsilon.
\end{eqnarray}
Thus, $C_{MC}(I_1,\ldots,I_L;F,N) \leq C_{inf}+\epsilon$.
By the definition of $C_{MC}^*(I_1,\ldots,I_L;F,N)$, we have
\begin{eqnarray}
\nonumber  C_{MC}^*(I_1,\ldots,I_L;F,N)\leq C_{MC}(I_1,\ldots,I_L;F,N) \leq C_{inf}+\epsilon   ~~for~\forall N\geq N^{*},
\end{eqnarray}
which implies that
\begin{eqnarray}
\nonumber  \lim\sup_{N\rightarrow\infty}C_{MC}^*(I_1,\ldots,I_L;F,N)\leq C_{inf}+\epsilon.
\end{eqnarray}
Since $\epsilon$ is arbitrary, we have
\begin{eqnarray}
\label{Eq_3_006} \lim\sup_{N\rightarrow\infty}C_{MC}^*(I_1,\ldots,I_L;F,N)\leq C_{inf}.
\end{eqnarray}

On the other hand, suppose that
\begin{eqnarray}
\nonumber \lim\inf_{N\rightarrow\infty}C_{MC}^*(I_1,\ldots,I_L;F,N)< C_{inf}.
\end{eqnarray}
Then there would be a positive constant $\tau > 0$, and a sequence $\{N_{k}\}$ such that $N_{k}\rightarrow\infty$, and
\begin{eqnarray}
\label{Eq_3_008} C_{MC}^*(I_1,\ldots,I_L;F,N_{k})< C_{inf}-\tau.
\end{eqnarray}
For every such $C_{MC}^*(I_1,\ldots,I_L;F,N_{k})$, there must be a set of $(I_1^{k},I_2^{k}\ldots,I_L^{k})$ such that
\begin{eqnarray}
\nonumber C_{MC}^*(I_1,\ldots,I_L;F,N_{k})= C_{MC}(I_1^{k},I_2^{k}\ldots,I_L^{k};F,N_{k}).
\end{eqnarray}
Using the inequality (\ref{Eq_3_0001}) and (\ref{Eq_3_008}), for large enough $K$, we have $\frac{\delta}{\sqrt{N_{K}}}<\tau$,
\begin{eqnarray}
\nonumber C(I_1^{K},I_2^{K}\ldots,I_L^{K};F)\leq C_{MC}^*(I_1,\ldots,I_L;F,N_{K})+\tau < C_{inf},
\end{eqnarray}
which contradicts the definition of $C_{inf}$. Therefore,
\begin{eqnarray}
\label{Eq_3_011} \lim\inf_{N\rightarrow\infty}C_{MC}^*(I_1,\ldots,I_L;F,N) \geq C_{inf}.
\end{eqnarray}

By the inequality (\ref{Eq_3_006}) and (\ref{Eq_3_011}),
\begin{eqnarray}
\nonumber C_{inf}\leq \lim\inf_{N\rightarrow\infty}C_{MC}^*(I_1,\ldots,I_L;F,N)\leq \lim\sup_{N\rightarrow\infty}C_{MC}^*(I_1,\ldots,I_L;F,N)\leq C_{inf},
\end{eqnarray}
which implies that (\ref{Eq_3_001}).
\end{proof}

\begin{remark}
The assumption (\ref{Eq_3_0001}) is not restrictive, since, by the central limit theorem, the error term of this Monte Carlo approximation is $O(N^{-1/2})$ regardless of the dimensionality of $Y$ (see \cite{Liu01}).
\end{remark}

Secondly, we derive the necessary conditions for optimal sensor decision rules in the sense of minimizing $C_{MC}(I_1(y_1),$ $\ldots,I_L(y_L);F,N)$ for a parallel distributed detection system.%  as follows:
\begin{theorem}\label{thm_1}
If $\{I_1(y_1),\ldots,I_{L}(y_L)\}$ are a set of optimal sensor decision rules which minimize $C_{MC}(I_1(y_1),$ $\ldots,I_L(y_L);F,N)$ in (\ref{Eq_1_17}) in a parallel distributed Bayesian  detection fusion system, then $\{I_1(y_1),\ldots,I_{L}(y_L)\}$ must satisfy the following equations: %$\{(I_1(Y_{1i}),\ldots,I_{L}(Y_{Li})), i=1,\ldots,N\}$
\begin{eqnarray}
\nonumber I_1(Y_{1i})&=&I[P_{11}(I_2(Y_{2i}),I_3(Y_{3i}),\ldots,I_L(Y_{Li});F)\\
\label{Eq_3_1}&&\cdot\hat{L}(Y_{1i},Y_{2i},\ldots,Y_{Li})], \qquad for ~ i=1,\ldots,N \\
\nonumber I_2(Y_{2i})&=&I[P_{21}(I_1(Y_{1i}),I_3(Y_{3i}),\ldots,I_L(Y_{Li});F)\\
\label{Eq_3_2}&&\cdot\hat{L}(Y_{1i},Y_{2i},\ldots,Y_{Li})], \qquad for ~ i=1,\ldots,N \\
\nonumber    & \cdots\cdots \\
\nonumber I_L(Y_{Li})&=&I[P_{L1}(I_1(Y_{1i}),(I_2(Y_{2i})),\ldots,I_{L-1}(Y_{(L-1)i});F)\\
\label{Eq_3_3}&&\cdot\hat{L}(Y_{1i},Y_{2i},\ldots,Y_{Li})], \qquad for ~ i=1,\ldots,N
\end{eqnarray}
where $P_{j1}(\cdot), j=1,\ldots,L$ are defined by (\ref{Eq_1_12}), $I[\cdot]$ is an indicator function denoted as follows:\\
\begin{eqnarray}
\label{Eq_3_4}I[x]&=&\left\{
                      \begin{array}{ll}
                       1, & \hbox{ if $x\geq0$;} \\
                        0, & \hbox{ if $ x<0$.}
                      \end{array}
                    \right.
\end{eqnarray}
\end{theorem}

\begin{proof}
Since both $P_{j1}(\cdot)$ and $P_{j2}(\cdot)$ are independent of $I_j(y_j)$ for $j=1,\ldots,L$, if $I_1(y_1)$  minimizes the Monte Carlo approximation of (\ref{Eq_1_17}), then $I_1(Y_{1i})$ should be equal to 1 when $P_{11}(I_2(Y_{2i}),I_3(Y_{3i}),$ $\ldots,I_L(Y_{Li});$ $F)\hat{L}(Y_{1i},Y_{2i},\ldots,Y_{Li})$ is positive for $i=1,\ldots,N$, otherwise it should be equal to 0. Thus, we have (\ref{Eq_3_1}) by the definition of $I[x]$ in (\ref{Eq_3_4}). Similarly, by (\ref{Eq_1_20}), we have (\ref{Eq_3_2})--(\ref{Eq_3_3}).
\end{proof}

%Note that, in the following sections, we assume that the samples drawn from the trial distribution have been fixed so that $C_{MC}(I_1(y_1),\ldots,I_L(y_L);F,N)$  has no randomness.

\section{Monte Carlo Gauss-Seidel Iterative Algorithm Its Convergence} \label{sec_4}
\subsection{Monte Carlo Gauss-Seidel Iterative Algorithm} \label{sec_4_1}

Let the sensor decision rules at the $k$th stage of iteration  be denoted by $\{(I_1^{k}(Y_{1i}),\ldots,I_{L}^{k}(Y_{Li})), i=1,\ldots,N\}$ with the initial set $\{(I_1^{0}(Y_{1i}),\ldots,I_{L}^{0}(Y_{Li})), i=1,\ldots,N\}$. Suppose the fusion rule is fixed. Based on Theorem \ref{thm_1}, we can drive a Gauss-Seidel iterative algorithm for minimizing $C_{MC}(I_1(y_1),$ $\ldots,I_L(y_L);F,N)$ in (\ref{Eq_1_18}) as follows.

\begin{algorithm}[Monte Carlo Gauss-Seidel iterative algorithm]\label{alg_1}
~\\
\begin{itemize}
           \item  Step 1: Draw samples $Y_{1},\ldots,Y_{N}$ from an importance density $g(y_1,y_2,\ldots,y_L)$.
           \item  Step 2: Given a fusion rule $F$ and initialize $L$ sensor decision rules $j=1,\ldots,L$,
\begin{eqnarray}
\label{Eq 4_1}  I_j^{0}(Y_{ji})=0/1 \qquad for \qquad i=1,\ldots,N.
\end{eqnarray}
           \item
Step 3: Iteratively search $L$ sensor decision rules for better system performance until a terminate criterion step 4 is satisfied. The $(k+1)$th stage of the iteration is as follows:
\begin{eqnarray}
\nonumber I_1^{k+1}(Y_{1i})&=&I[P_{11}(I_2^{k}(Y_{2i}),I_3^{k}(Y_{3i}),\ldots,I_L^{k}(Y_{Li});F)\\
\label{Eq 4_2} &&\cdot\hat{L}(Y_{1i},Y_{2i},\ldots,Y_{Li})], \qquad for ~ i=1,\ldots,N \\
\nonumber I_2^{k+1}(Y_{2i})&=&I[P_{21}(I_1^{k+1}(Y_{2i}),I_3^{k}(Y_{3i}),\ldots,I_L^{k}(Y_{Li});F)\\
\label{Eq 4_3} &&\cdot\hat{L}(Y_{1i},Y_{2i},\ldots,Y_{Li})], \qquad for ~ i=1,\ldots,N \\
\nonumber    & \cdots\cdots \\
\nonumber I_L^{k+1}(Y_{Li})&=&I[P_{L1}(I_1^{k+1}(Y_{1i}),(I_2^{k+1}(Y_{2i})),\ldots,I_{L-1}^{k+1}(Y_{(L-1)i});F)\\
\label{Eq 4_4} &&\cdot\hat{L}(Y_{1i},Y_{2i},\ldots,Y_{Li})], \qquad for ~ i=1,\ldots,N.
\end{eqnarray}
           \item Step 4: A termination criterion of the iteration process is, for $i=1,\ldots,N$
\begin{eqnarray}
\nonumber I_1^{k+1}(Y_{1i})&=&I_1^{k}(Y_{1i}) ,\\
\nonumber I_2^{k+1}(Y_{2i})&=&I_2^{k}(Y_{2i}) ,\\
\nonumber &\cdots\cdots \\
\label{Eq 4_5}  I_L^{k+1}(Y_{Li})&=&I_L^{k}(Y_{Li}).
\end{eqnarray}
\end{itemize}
\end{algorithm}

\begin{remark}\label{1}%%¶¨ÒåΪ y1 »¹ÊÇ y ÓëY_{1i}µÄ¾àÀ룬»òÕß˵ÊÇͶӰµÄ¾àÀë.
Once we obtain $I_1(Y_{1i})$ for $i=1,\ldots, N$, then $I_1(y)$ can be obtained by defining  $I_1(y_1)=I_1(Y_{1i})$  when the distance $||y_1-Y_{1i}||$ is less than $||y_1-Y_{1j}||$, for all $j\neq i$. Similarly, we can obtain $I_i(y_i)$ for $i=2,\ldots, L$.
%\textcolor[rgb]{0.98,0.00,0.00}{Once we obtain $I_1(Y_{1i})$ for $i=1,\ldots, N$, then $I_1(y_1)$ can be obtained by defining  $I_1(y_1)=I_1(Y_{1i})$  when the distance $||y_1-Y_{1i}||$ is less than $||y_1-Y_{1j}||$, for all $j\neq i$, where $y=(y_1,y_2,\cdots,y_L)$. Similarly, we can obtain $I_i(y_i)$ for $i=2,\ldots, L$.}
\end{remark}

\begin{remark}\label{2}
The main computation burden of Algorithm \ref{alg_1} is in (\ref{Eq 4_2})--(\ref{Eq 4_4}). If the number of discretized points $N_1=N_2=\ldots=N_L=N$ in (10) of \cite{Zhu-Blum-Luo-Wong00}, then $P_{j1}(\cdot)\hat{L}(Y_{1i},Y_{2i},\ldots,Y_{Li}), j=1,\ldots,L, i=1,\ldots,N$ are computed $L\times N$ times in Algorithm \ref{alg_1}. However, in \cite{Zhu-Blum-Luo-Wong00}, they are computed $LN^L$ times. In next section, we prove Algorithm \ref{alg_1} terminates in finite steps. Thus, the computation complexity of Algorithm \ref{alg_1} is $O(LN)$ compared with $O(LN^L)$ of the algorithm in \cite{Zhu-Blum-Luo-Wong00}.
\end{remark}

\subsection{Convergence of Monte Carlo Gauss-Seidel Iterative Algorithm} \label{sec_4_2}
Now we prove that Algorithm \ref{alg_1} must converge to a local optimal value and the algorithm cannot oscillate infinitely often, i.e., terminate after a finite number of iterations.

For convenience, for $j=1,\ldots,L$, we denote $C_{MC}$ (\ref{Eq_1_19})--(\ref{Eq_1_20}) in the $(k+1)$th iteration process by
%\begin{eqnarray}
%\nonumber  &&C_{MC}(I_1^{k},I_2^{k},\ldots,I_L^{k})\\
%\label{Eq 4_6} &=&\frac{1}{N}\sum_{i=1}^{N}I_{\Omega_0}(I_1^{k}(Y_{1i}),I_2^{k}(Y_{2i}),\ldots,I_{L}^{k}(Y_{Li}))\frac{ \hat{L}(Y_{1i},Y_{2i},\ldots,Y_{Li})}{g(Y_{1i},Y_{2i},\ldots,Y_{Li})}+c \\[3mm]
%\nonumber  &&G_{j}^{k+1}(Y_{ji})\\
%\nonumber  &=&P_{j1}(I_1^{k+1}(Y_{1i}),\ldots,I_{j-1}^{k+1}(Y_{(j-1)i}),I_{j+1}^{k}(Y_{(j+1)i}),\ldots,I_{L}^{k}(Y_{Li}))\\
%\label{Eq 4_7} &&\cdot \hat{L}(Y_{1i},Y_{2i},\ldots,Y_{Li}),\qquad for ~j=1,\ldots,L \qquad i=1,\ldots,N
%\end{eqnarray}

%\begin{eqnarray}
%\nonumber  &&C_{MC}(I_1^{k+1},\ldots,I_{j}^{k+1},I_{j+1}^{k},\ldots,I_L^{k})\\
%\nonumber  &=&\frac{1}{N}\sum_{i=1}^{N}I_{\Omega_0}(I_1^{k+1}(Y_{1i}),\ldots,I_{j}^{k+1}(Y_{ji}),I_{j+1}^{k}(Y_{(j+1)i}),\ldots,I_{L}^{k}(Y_{Li}))\\
%\label{Eq 4_6}  &&\cdot\frac{ \hat{L}(Y_{1i},Y_{2i},\ldots,Y_{Li})}{g(Y_{1i},Y_{2i},\ldots,Y_{Li})}+c \qquad for ~j=1,\ldots,L \textcolor[rgb]{0.00,0.00,1.00}{\qquad i=1,\ldots,N}
%\end{eqnarray}
\begin{eqnarray}
\nonumber&&C_{MC}(I_1^{k+1},\ldots,I_j^{k+1},I_{j+1}^{k},\ldots,I_L^{k};F,N)\\
\nonumber&=&c+\frac{1}{N}\sum_{i=1}^{N}\{[1-I_j^{k+1}(Y_{ji})]P_{j1}(I_1^{k+1}(Y_{1i}),\ldots,I_{j-1}^{k+1}(Y_{(j-1)i}),I_{j+1}^{k}(Y_{(j+1)i}),\ldots,I_{L}^{k}(Y_{Li});F,N)\\
\label{Eq 4_6}
&&+P_{j2}(I_1^{k+1}(Y_{1i}),\ldots,I_{j-1}^{k+1}(Y_{(j-1)i}),I_{j+1}^{k}(Y_{(j+1)i}),\ldots,I_{L}^{k}(Y_{Li});F,N)\}\frac{\hat{L}(Y_{1i},Y_{2i},\ldots,Y_{Li})}{g(Y_{1i},Y_{2i},\ldots,Y_{Li})}. \end{eqnarray}
%% $I_{\Omega_0}$ ²»ÐèÒª¶¨Òå
%\textcolor[rgb]{0.98,0.00,0.00}{where $I_{\Omega_0}(\cdot)$ is defined by (\ref{Eq_1_10})--(\ref{Eq_1_11}) and } %ɾ³ý
Similarly, we denote the $(k+1)$th iteration process of the iterative items $P_{j1}(\cdot)\hat{L}(\cdot)$ in (\ref{Eq 4_2})--(\ref{Eq 4_4}) by
\begin{eqnarray}
\nonumber  G_{j}^{k+1}(Y_{ji})
\nonumber  &=&P_{j1}(I_1^{k+1}(Y_{1i}),\ldots,I_{j-1}^{k+1}(Y_{(j-1)i}),I_{j+1}^{k}(Y_{(j+1)i}),\ldots,I_{L}^{k}(Y_{Li});F,N)\\
\label{Eq 4_7} &&\cdot \hat{L}(Y_{1i},Y_{2i},\ldots,Y_{Li}), \qquad for ~i=1,\ldots,N, ~ j=1,\ldots,L.
\end{eqnarray}

\begin{lemma}\label{lem_1}
$C_{MC}(I_1^{k+1},\ldots,I_j^{k+1},I_{j+1}^{k},\ldots,I_L^{k};F,N)$ is non-increasing as $j$ is increased and $C_{MC}(I_1^{k+1},$ $I_2^{k+1},\ldots,I_L^{k+1};F,N)\leq C_{MC}(I_1^{k},I_2^{k},\ldots,I_L^{k};F,N)$.
\end{lemma}
\begin{proof}
Using (\ref{Eq 4_6})-(\ref{Eq 4_7}), we have

\begin{eqnarray}
\nonumber&&C_{MC}(I_1^{k+1},\ldots,I_j^{k+1},I_{j+1}^{k},\ldots,I_L^{k};F,N)=\frac{1}{N}\sum_{i=1}^{N}\frac{[1-I_j^{k+1}(Y_{ji})]}{g(Y_{1i},Y_{2i},\ldots,Y_{Li})}G_{j}^{k+1}(Y_{ji})+C_{j},
\end{eqnarray}
where
\begin{eqnarray}
\nonumber C_{j}=\frac{1}{N}\sum_{i=1}^{N}P_{j2}(I_1^{k+1}(Y_{1i}),\ldots,I_{j-1}^{k+1}(Y_{(j-1)i}),I_{j+1}^{k}(Y_{(j+1)i}),\ldots,I_{L}^{k}(Y_{Li});F,N)\frac{\hat{L}(Y_{1i},Y_{2i},\ldots,Y_{Li})}{g(Y_{1i},Y_{2i},\ldots,Y_{Li})}+c
\end{eqnarray}
is a constant independent of $I_j^{k}$ and $I_j^{k+1}$.
\begin{eqnarray}
\nonumber&&C_{MC}(I_1^{k+1},\ldots,I_j^{k+1},I_{j+1}^{k},\ldots,I_L^{k};F,N) \\
\nonumber&=&\frac{1}{N}\sum_{i=1}^{N}\frac{[1-I_j^{k}(Y_{ji})]+[I_j^{k}(Y_{ji})-I_j^{k+1}(Y_{ji})]}{g(Y_{1i},Y_{2i},\ldots,Y_{Li})}G_{j}^{k+1}(Y_{ji})+C_{j}\\
\nonumber&=&\frac{1}{N}\sum_{i=1}^{N}\frac{[1-I_j^{k}(Y_{ji})]}{g(Y_{1i},Y_{2i},\ldots,Y_{Li})}G_{j}^{k+1}(Y_{ji})+C_{j}+\frac{1}{N}\sum_{i=1}^{N}\frac{[I_j^{k}(Y_{ji})-I_j^{k+1}(Y_{ji})]}{g(Y_{1i},Y_{2i},\ldots,Y_{Li})}G_{j}^{k+1}(Y_{ji})\\
%\nonumber&=&C_{MC}(I_1^{k+1},\ldots,I_{j-1}^{k+1},I_{j}^{k},\ldots,I_L^{k};F,N)+D_{j}^{k+1},
\label{Eq 04_07}&=&C_{MC}(I_1^{k+1},\ldots,I_{j-1}^{k+1},I_{j}^{k},\ldots,I_L^{k};F,N)+D_{j}^{k+1},
\end{eqnarray}
where
\begin{eqnarray}
%\nonumber &&D_{j}^{k+1}=\frac{1}{N}\sum_{i=1}^{N}\frac{[I_j^{k}(Y_{ji})-I_j^{k+1}(Y_{ji})]G_{j}^{k+1}(Y_{ji})}{g(Y_{1i},Y_{2i},\ldots,Y_{Li})}.
\label{Eq 04_08} &&D_{j}^{k+1}=\frac{1}{N}\sum_{i=1}^{N}\frac{[I_j^{k}(Y_{ji})-I_j^{k+1}(Y_{ji})]G_{j}^{k+1}(Y_{ji})}{g(Y_{1i},Y_{2i},\ldots,Y_{Li})}.
\end{eqnarray}
%The last inequality holds due to the fact that (\ref{Eq 4_2}-\ref{Eq 4_4}) implies
%$I_{j}^{k+1}(Y_{ji})=0$ if and only if $G_{j}^{k+1}(Y_{ji})<0$ and
%$I_{j}^{k+1}(Y_{ji})=1$ if and only if $G_{j}^{k+1}(Y_{ji})\geq0$  for~$j=1,\ldots,L$,
%$i=1,\ldots,N$. That is to say $[I_j^{k}(Y_{ji})-I_j^{k+1}(Y_{ji})]G_{j}^{k+1}(Y_{ji})\leq0$, then all terms of the summation $D_{j}^{k+1}\leq 0$ because $g(\cdot)$ is a trial distribution (i.e. $g(\cdot)\geq0$ ).\\
Note that (\ref{Eq 4_2})-(\ref{Eq 4_4}) implie that $I_{j}^{k+1}(Y_{ji})=0$ if and only if $G_{j}^{k+1}(Y_{ji})<0$ and $I_{j}^{k+1}(Y_{ji})=1$ if and only if $G_{j}^{k+1}(Y_{ji})\geq0$ for $i=1,\ldots,N,~j=1,\ldots,L$.
That is to say
\begin{eqnarray}
\label{Eq 04_09}[I_j^{k}(Y_{ji})-I_j^{k+1}(Y_{ji})]G_{j}^{k+1}(Y_{ji})\leq0.
\end{eqnarray}
Thus, for $\forall i,j,k$
\begin{eqnarray}
\label{Eq 04_10} \frac{[I_j^{k}(Y_{ji})-I_j^{k+1}(Y_{ji})]G_{j}^{k+1}(Y_{ji})}{g(Y_{1i},Y_{2i},\ldots,Y_{Li})}\leq0,
\end{eqnarray}
the inequality holds because $g(\cdot)$ is a trial distribution and well-defined (i.e., $g(\cdot)>0$). Then the summation of all terms  $D_{j}^{k+1}\leq 0$.
%\textcolor[rgb]{0.98,0.00,0.00}{then all terms of the summation $D_{j}^{k+1}\leq 0$ because $g(\cdot)$ is a trial distribution and well-defined (i.e., $g(\cdot)>0$).}
Thus, for $\forall j\leq L$,
$C_{MC}(I_1^{k+1},\ldots,I_j^{k+1},I_{j+1}^{k},\ldots,I_L^{k};F,N)\leq C_{MC}(I_1^{k+1},\ldots,I_{j-1}^{k+1},$ $I_{j}^{k},\ldots,I_L^{k};F,N)$, $C_{MC}(I_1^{k+1},I_2^{k+1},\ldots,I_L^{k+1};F,N)\leq C_{MC}(I_1^{k},I_2^{k},\ldots,I_L^{k};F,N)$.
\end{proof}
Note that $C_{MC}(I_1^{k},I_2^{k},\ldots,I_L^{k};F,N)$ is finite valued. From Lemma \ref{lem_1}, it must converge to a stationary point after a finite number of iterations.
\begin{theorem}\label{thm_2}
The $I_1^{k},I_2^{k},\ldots,I_L^{k}$ are  finitely convergent.
\end{theorem}
\begin{proof}
By Lemma \ref{lem_1}, $C_{MC}(I_1^{k},I_2^{k},\ldots,I_L^{k};F,N)$ must converge to a stationary point after a finite number of iterations, i.~e.
\begin{eqnarray}
\label{Eq 04_11} &&C_{MC}(I_1^{k+1},\ldots,I_j^{k+1},I_{j+1}^{k},\ldots,I_L^{k};F,N)=C_{MC}(I_1^{k+1},\ldots,I_{j-1}^{k+1},I_{j}^{k},\ldots,I_L^{k};F,N).
\end{eqnarray}
Using (\ref{Eq 04_07}) and (\ref{Eq 04_11}), we can derive that $D_{j}^{k+1}=0$. Combine (\ref{Eq 04_08})-(\ref{Eq 04_10}),
\begin{eqnarray}
%\nonumber &&\frac{1}{N}\sum_{i=1}^{N}\frac{[I_j^{k}(Y_{ji})-I_j^{k+1}(Y_{ji})]G_{j}^{k+1}(Y_{ji})}{g(Y_{1i},Y_{2i},\ldots,Y_{Li})}=0,\\
\nonumber &&[I_j^{k}(Y_{ji})-I_j^{k+1}(Y_{ji})]G_{j}^{k+1}(Y_{ji})=0 \qquad for \qquad i=1,\ldots,N,
\end{eqnarray}
which implies either
\begin{eqnarray}
\nonumber &&I_j^{k}(Y_{ji})-I_j^{k+1}(Y_{ji})=0, \qquad i.e. \qquad I_j^{k}(Y_{ji})=I_j^{k+1}(Y_{ji})
\end{eqnarray}
or
\begin{eqnarray}
\nonumber &&G_{j}^{k+1}(Y_{ji})=0, \qquad i.e. \qquad I_j^{k+1}(Y_{ji})=1.
\end{eqnarray}
It follows that when $C_{MC}$ converges to a stationary point, either $I_j^{k+1}(Y_{ji})$ is invariant, or $I_j^{k+1}(Y_{ji})=1,I_j^{k}(Y_{ji})=0$. That is $I_j^{k+1}(Y_{ji})$ can only change from 0 to 1 at most a finite number of times. Thus the algorithm often cannot oscillate infinitely.
\end{proof}

\begin{theorem}\label{thm_3}
For the fixed AND fusion rule, $(I_1(y_{1}),I_2(y_{2}),\ldots,I_L(y_{L}))$  minimize the Monte Carlo cost function (\ref{Eq_1_18}) if and only if they satisfy the following equations:
\begin{eqnarray}
\label{Eq 4_8}&&I_1(Y_{1i})\cdot I_2(Y_{2i})\cdot\cdots\cdot I_L(Y_{Li})=1 ~~ if ~ \hat{L}(Y_{1i},Y_{2i},\ldots,Y_{Li})\geq0 \quad for\quad i=1,\cdots,N\\
\label{Eq 4_9}&&I_1(Y_{1i})\cdot I_2(Y_{2i})\cdot\cdots\cdot I_L(Y_{Li})=0 ~~ if ~ \hat{L}(Y_{1i},Y_{2i},\ldots,Y_{Li})<0 \quad for\quad i=1,\cdots,N.
\end{eqnarray}
Moreover, one of the optimal solutions is
\begin{eqnarray}
\label{Eq 4_08}&&I_1(Y_{1i})=I_2(Y_{2i})=\cdots=I_L(Y_{Li})=1 ~ ~if ~ \hat{L}(Y_{1i},Y_{2i},\ldots,Y_{Li})\geq0 \quad for\quad i=1,\cdots,N,\\
\label{Eq 4_09}&&I_1(Y_{1i})=I_2(Y_{2i})=\cdots=I_L(Y_{Li})=0 ~ ~if ~\hat{L}(Y_{1i},Y_{2i},\ldots,Y_{Li})<0 \quad for\quad i=1,\cdots,N.
\end{eqnarray}
\end{theorem}

\begin{proof}
%$C_{MC}(I_1^{k},I_2^{k},\ldots,I_L^{k})$ must converge to a stationary point after a finite number of iterations.
For the fixed AND fusion rule, %$I_{\Omega_0}(Y_{1i},Y_{2i},\ldots,Y_{Li})=1-I_1(Y_{1i})\cdot I_2(Y_{2i})\cdot\cdots I_L(Y_{Li})$. Thus,
\begin{eqnarray}
\label{Eq 04_01} I_{\Omega_0}(Y_{1i},Y_{2i},\ldots,Y_{Li})&=&1-I_1(Y_{1i})\cdot I_2(Y_{2i})\cdot\cdots I_L(Y_{Li}).
\end{eqnarray}
Substituting (\ref{Eq 04_01}) into (\ref{Eq_1_17}) and simplifying, we have
\begin{eqnarray}
\nonumber &&C_{MC}(I_1(y_1),\ldots,I_L(y_L);F,N)\\
%\nonumber &=&\frac{1}{N}\sum_{i=1}^{N}I_{\Omega_0}(Y_{1i},Y_{2i},\ldots,Y_{Li})\frac{\hat{L}(Y_{1i},Y_{2i},\ldots,Y_{Li})}{g(Y_{1i},Y_{2i},\ldots,Y_{Li})}+c\\
\nonumber &=&\frac{1}{N}\sum_{i=1}^{N}\{1-I_1(Y_{1i})\cdot I_2(Y_{2i})\cdot\cdots I_L(Y_{Li})\}\cdot\frac{\hat{L}(Y_{1i},Y_{2i},\ldots,Y_{Li})}{g(Y_{1i},Y_{2i},\ldots,Y_{Li})}+c\\
\nonumber &=&\frac{1}{N}\sum_{i=1}^{N}\frac{\hat{L}(Y_{1i},Y_{2i},\ldots,Y_{Li})}{g(Y_{1i},Y_{2i},\ldots,Y_{Li})}-\frac{1}{N}\sum_{i=1}^{N}I_1(Y_{1i})\cdot I_2(Y_{2i})\cdot \cdots I_L(Y_{Li})\cdot\frac{\hat{L}(Y_{1i},Y_{2i},\ldots,Y_{Li})}{g(Y_{1i},Y_{2i},\ldots,Y_{Li})}+c\\
\nonumber &=&C_0-\tilde{C}_{MC}(I_1(y_1),\ldots,I_L(y_L);F,N)
\end{eqnarray}
%If we denote the second part of the last equation by
where
\begin{eqnarray}
\nonumber C_0&=&\frac{1}{N}\sum_{i=1}^{N}\frac{\hat{L}(Y_{1i},Y_{2i},\ldots,Y_{Li})}{g(Y_{1i},Y_{2i},\ldots,Y_{Li})}+c,\\
\label{Eq 4_10}\tilde{C}_{MC}(I_1(y_1),\ldots,I_L(y_L);F,N)&=&\frac{1}{N}\sum_{i=1}^{N}I_1(Y_{1i})\cdot I_2(Y_{2i})\cdot \cdots I_L(Y_{Li})\cdot\frac{\hat{L}(Y_{1i},Y_{2i},\ldots,Y_{Li})}{g(Y_{1i},Y_{2i},\ldots,Y_{Li})}.
\end{eqnarray}
$C_0$ is a constant, then minimizing $C_{MC}(I_1(y_1),\ldots,I_L(y_L);F,N)$ is equivalent to maximize $\tilde{C}_{MC}(I_1(y_1),$ $\ldots,I_L(y_L);F,N)$. Note that $g(Y_{1i},Y_{2i},\ldots,Y_{Li})>0$ and $I_1(Y_{1i})\cdot I_2(Y_{2i})\cdot\cdots\cdot I_L(Y_{Li})=1$ or 0.  For arbitrary $Y_i=(Y_{1i},Y_{2i},\cdots,Y_{Li})$, $(I_1(y_{1}),I_2(y_{2}),\ldots,I_L(y_{L}))$ maximize $\tilde{C}_{MC}$ if and only if they satisfy the following equations:
\begin{eqnarray}
I_1(Y_{1i})\cdot I_2(Y_{2i})\cdot\cdots\cdot I_L(Y_{Li})=1 ~ if ~ \hat{L}(Y_{1i},Y_{2i},\ldots,Y_{Li})\geq0\quad for\quad i=1,\cdots,N\nonumber\\
I_1(Y_{1i})\cdot I_2(Y_{2i})\cdot\cdots\cdot I_L(Y_{Li})=0 ~ if ~ \hat{L}(Y_{1i},Y_{2i},\ldots,Y_{Li})<0\quad for\quad i=1,\cdots,N\nonumber
\end{eqnarray}
Thus, we have (\ref{Eq 4_8})--(\ref{Eq 4_9}).
\end{proof}

\begin{theorem}\label{thm_4}
For the fixed OR fusion rule, $(I_1(y_{1}),I_2(y_{2}),\ldots,I_L(y_{L}))$ minimize the Monte Carlo cost function (\ref{Eq_1_18}) if and only if they satisfy the following equations:
\begin{eqnarray}
\label{Eq 4_11} (1-I_1(Y_{1i}))\cdot (1-I_2(Y_{2i}))\cdot\cdots\cdot (1-I_L(Y_{Li}))=0 ~ if ~ \hat{L}(Y_{1i},Y_{2i},\ldots,Y_{Li})\geq0 ~ for ~ i=1,\cdots,N\\
\label{Eq 4_12} (1-I_1(Y_{1i}))\cdot (1-I_2(Y_{2i}))\cdot\cdots\cdot (1-I_L(Y_{Li}))=1 ~ if ~ \hat{L}(Y_{1i},Y_{2i},\ldots,Y_{Li})<0 ~ for ~ i=1,\cdots,N.
\end{eqnarray}
Moreover, one of the optimal solutions is
\begin{eqnarray}
\label{Eq 4_13}&&I_1(Y_{1i})=I_2(Y_{2i})=\cdots=I_L(Y_{Li})=1 ~ ~if ~ \hat{L}(Y_{1i},Y_{2i},\ldots,Y_{Li})\geq0 \quad for\quad i=1,\cdots,N\\
\label{Eq 4_14}&&I_1(Y_{1i})=I_2(Y_{2i})=\cdots=I_L(Y_{Li})=0 ~ ~if ~\hat{L}(Y_{1i},Y_{2i},\ldots,Y_{Li})<0 \quad for\quad i=1,\cdots,N.
\end{eqnarray}
\end{theorem}
\begin{proof}
For the fixed OR fusion rule,
\begin{eqnarray}
\label{Eq 04_02}I_{\Omega_0}(Y_{1i},Y_{2i},\ldots,Y_{Li})&=&(1-I_1(Y_{1i}))\cdot (1-I_2(Y_{2i}))\cdot\cdots (1-I_L(Y_{Li})).
\end{eqnarray}
%Thus,
Substituting (\ref{Eq 04_02}) into (\ref{Eq_1_17}), we have
\begin{eqnarray}
\nonumber &&C_{MC}(I_1(y_1),\ldots,I_L(y_L);F,N)\\
%\nonumber &=&\frac{1}{N}\sum_{i=1}^{N}I_{\Omega_0}(Y_{1i},Y_{2i},\ldots,Y_{Li})\frac{\hat{L}(Y_{1i},Y_{2i},\ldots,Y_{Li})}{g(Y_{1i},Y_{2i},\ldots,Y_{Li})}+c\\
\nonumber &=&\frac{1}{N}\sum_{i=1}^{N}\{(1-I_1(Y_{1i}))\cdot (1-I_2(Y_{2i}))\cdot\cdots (1-I_L(Y_{Li}))\}\cdot\frac{\hat{L}(Y_{1i},Y_{2i},\ldots,Y_{Li})}{g(Y_{1i},Y_{2i},\ldots,Y_{Li})}+c.
%\nonumber &=&\frac{1}{N}\sum_{i=1}^{N}\frac{\hat{L}(Y_{1i},Y_{2i},\ldots,Y_{Li})}{g(Y_{1i},Y_{2i},\ldots,Y_{Li})}-\frac{1}{N}\sum_{i=1}^{N}I_1(Y_{1i})\cdot I_2(Y_{2i})\cdot \cdots I_L(Y_{Li})\cdot\frac{\hat{L}(Y_{1i},Y_{2i},\ldots,Y_{Li})}{g(Y_{1i},Y_{2i},\ldots,Y_{Li})}+c
\end{eqnarray}
%\begin{eqnarray}
%\label{Eq 26}C1_{MC}(I_1(y_1),\ldots,I_L(y_L);F)&=&\frac{1}{N}\sum_{i=1}^{N}I_1(Y_{1i})\cdot I_2(Y_{2i})\cdot \cdots I_L(Y_{Li})\cdot\frac{\hat{L}(Y_{1i},Y_{2i},\ldots,Y_{Li})}{g(Y_{1i},Y_{2i},\ldots,Y_{Li})}
%\end{eqnarray}
Since $c$ is a constant, $(1-I_1(Y_{1i}))\cdot (1-I_2(Y_{2i}))\cdot\cdots\cdot (1-I_L(Y_{Li}))=0$ or 1 and $g(Y_{1i},Y_{2i},\ldots,Y_{Li})>0$,  $(I_1(y_{1}),I_2(y_{2}),\ldots,I_L(y_{L}))$ minimize $C_{MC}(I_1(y_1),\ldots,I_L(y_L);F,N)$ if and only if they satisfy the following equations:
\begin{eqnarray}
\nonumber (1-I_1(Y_{1i}))\cdot (1-I_2(Y_{2i}))\cdot\cdots\cdot (1-I_L(Y_{Li}))=0 ~~ if ~ \hat{L}(Y_{1i},Y_{2i},\ldots,Y_{Li})\geq0 ~~ for ~ i=1,\cdots,N\\
\nonumber (1-I_1(Y_{1i}))\cdot (1-I_2(Y_{2i}))\cdot\cdots\cdot (1-I_L(Y_{Li}))=1 ~~ if ~ \hat{L}(Y_{1i},Y_{2i},\ldots,Y_{Li})<0 ~~ for ~ i=1,\cdots,N
\end{eqnarray}
Thus, we have (\ref{Eq 4_11})--(\ref{Eq 4_12}).
\end{proof}

\begin{theorem}\label{cor_4}
For the fixed  AND fusion rule and any initial value, the solution of Monte Carlo  Gauss-Seidel iterative algorithm must converge to the analytically optimal solution given in Theorem \ref{thm_3}.
\end{theorem}
\begin{proof}
Without loss of generality, we assume that Monte Carlo Gauss-Seidel iterative algorithm terminated at $K$-th iteration for any initial value
and $(I_1^{K}(Y_{1i}),I_2^{K}(Y_{2i}),\ldots,I_{L}^{K}(Y_{Li}))$ is the set of L sensor decision rules at $K$-th iteration.
%We denote $(I_1^{K}(Y_{1i}),I_2^{K}(Y_{2i}),\ldots,I_{L}^{K}(Y_{Li}))$ by the solution of Monte Carlo Gauss-Seidel iterative algorithm (i.e. algorithm \ref{alg_1}.)
We need to prove that
\begin{eqnarray}
\label{Eq 4_15}I_1^{K}(Y_{1i})\cdot I_2^{K}(Y_{2i})\cdot\cdots\cdot I_L^{K}(Y_{Li})=1 \qquad if\qquad \hat{L}(Y_{1i},Y_{2i},\ldots,Y_{Li})\geq0\quad for\quad i=1,\cdots,N\\
\label{Eq 4_16}I_1^{K}(Y_{1i})\cdot I_2^{K}(Y_{2i})\cdot\cdots\cdot I_L^{K}(Y_{Li})=0 \qquad if\qquad \hat{L}(Y_{1i},Y_{2i},\ldots,Y_{Li})<0\quad for\quad i=1,\cdots,N.
\end{eqnarray}
Define two sets $S_{Y1}$ and $S_{Y0}$,
\begin{eqnarray}
\nonumber S_{Y1}&=&\{Y_{i}|\hat{L}(Y_{1i},Y_{2i},\ldots,Y_{Li})\geq0, i=1,\cdots,N\} \\
\nonumber S_{Y0}&=&\{Y_{i}|\hat{L}(Y_{1i},Y_{2i},\ldots,Y_{Li})<0, i=1,\cdots,N\}.
\end{eqnarray}

Firstly, we prove (\ref{Eq 4_15}) by a contradiction.
If there exists a sample $Y_{m}\in S_{Y1}$ such that $I_1^{K}(Y_{1m})\cdot I_2^{K}(Y_{2m})\cdot\cdots I_L^{K}(Y_{Lm})=0$, which implies that there must exist some $j$ such that $I_j^{K}(Y_{jm})=0$. For the fixed AND fusion rule, by (\ref{Eq_1_12}), for $j=1,\ldots,L$%%ÊÇ·ñÐèÒªÏêϸ˵Ã÷£¿
\begin{eqnarray}
\nonumber &&P_{j1}(I_1(Y_{1m}),\ldots,I_{j-1}(Y_{(j-1)m}),I_{j+1}(Y_{(j+1)m}),\ldots,I_{L}(Y_{Lm}))\\
\label{Eq 4_13} &=&I_1(Y_{1m})\cdot\cdots\cdot I_{j-1}(Y_{(j-1)m})\cdot I_{j+1}(Y_{(j+1)m})\cdot\cdots\cdot I_L(Y_{Lm}).
\end{eqnarray}
Thus, $P_{j1}(I_1(Y_{1m}),\ldots,I_{j-1}(Y_{(j-1)m}),I_{j+1}(Y_{(j+1)m}),\ldots,I_{L}(Y_{Lm}))=1$ or 0 for $j=1,\ldots,L$. Moreover, $P_{j1}(I_1^{K+1}(Y_{1m}),\ldots,I_{j-1}^{K+1}(Y_{(j-1)m}),I_{j+1}^{K}(Y_{(j+1)m}),\ldots,I_{L}^{K}(Y_{Lm}))=1$ or 0 for $j=1,\ldots,L$.
We can conclude that $P_{j1}(I_1^{K+1}(Y_{1m}),\ldots,I_{j-1}^{K+1}(Y_{(j-1)m}),I_{j+1}^{K}(Y_{(j+1)m}),\ldots,I_{L}^{K}(Y_{Lm}))\cdot\hat{L}(Y_{1m},Y_{2m},\ldots,$ $Y_{Lm})\geq0$ because of $Y_{m}\in S_{Y1}$, that is $\hat{L}(Y_{1m},Y_{2m},\ldots,Y_{Lm})\geq0$. By (\ref{Eq 4_2})-(\ref{Eq 4_4}), $I_j^{K+1}(Y_{jm})=1\neq I_j^{K}(Y_{jm})=0$.  It is a contradiction. Thus, we have (\ref{Eq 4_15}).

Secondly, we prove (\ref{Eq 4_16}) by a contradiction.
If there exists a sample $Y_{n}\in S_{Y0}$ such that $I_1^{K}(Y_{1n})\cdot I_2^{K}(Y_{2n})\cdot\cdots I_L^{K}(Y_{Ln})=1$, which implies that $I_j^{K}(Y_{jn})=1$ for $j=1,\ldots,L$. By (\ref{Eq 4_13}), we can conclude $P_{11}(I_{2}^{K}(Y_{2n}),\ldots,I_{L}^{K}(Y_{Ln}))=1$.
%Monte Carlo Guass-Seidel iterative algorithm combining with $Y_{n}\in S_{Y0}$ implies
Since $Y_{n}\in S_{Y0}$, that is $\hat{L}(Y_{1n},Y_{2n},\ldots,Y_{Ln})<0$. Thus,
\begin{eqnarray}
\nonumber P_{11}(I_{2}^{K}(Y_{2n}),\ldots,I_{L}^{K}(Y_{Ln}))\cdot\hat{L}(Y_{1n},Y_{2n},\ldots,Y_{Ln})<0\nonumber
\end{eqnarray}
By Algorithm \ref{alg_1} and (\ref{Eq_3_4}), we can conclude that $I_1^{K+1}(Y_{1n})=0\neq I_1^{K}(Y_{1n})=1$. It is a contradiction. Thus, we have (\ref{Eq 4_16}).
\end{proof}

\begin{remark}\label{Remark_5}
Since the OR fusion rule and AND fusion rule are dual each other, for the fixed OR fusion rule and any initial value, the solution of  Algorithm \ref{alg_1} must converge to the analytically optimal solution given in Theorem \ref{thm_4}.
\end{remark}

\section{Numerical Examples}\label{sec_5}

To evaluate the performance of the new algorithm, we investigate some examples
with large number of sensors where observation signal $s$  and observation noises are assumed  Gaussian and independent. Thus, the observations are dependent. %We also investigate some examples with large number of sensors. %
Since the previous distributed detection algorithm with general dependent observations does not work when the number of sensors is more than 5, we evaluate the new algorithm by comparing it with the centralized likelihood ratio method with 10 sensors and 100 sensors, respectively.

\subsection{Ten sensors}\label{sec_5_2}
%We compare the Monte Carlo iterative algorithm with Centralized algorithm.
We consider Monte Carlo importance sampling methods with AND, OR and 2 out of 5 (2/5) fusion rule.
\begin{example} Let us consider ten sensors model with observation signal $s$ and observation noises $v_{1}$, $v_{2}$, \ldots, $v_{10}$,
\begin{eqnarray}
\nonumber &&H_{1}:~~ y_{i}=s+v_{i}, ~~ for ~ i=1,\ldots,10\\
\nonumber &&H_{0}:~~ y_{i}=v_{i}, ~~ for ~ i=1,\ldots,10
\end{eqnarray}
where $s$, $v_{1}$, $v_{2}$ \ldots, $v_{10}$ are all mutually independent and
\begin{eqnarray}
%\nonumber &&s\sim N(1,2), ~~ v_{i}\sim N(0,3), ~~ for ~ i=1,\ldots,10
\nonumber &&s\sim N(1,0.4), ~~ v_{i}\sim N(0,0.6), ~~ for ~ i=1,\ldots,10
\end{eqnarray}
Therefore, the two conditional pdfs given $H_{0}$ and $H_{1}$ are
\begin{eqnarray}
\nonumber && p(y_1,y_2,\ldots,y_{10} |H_{1})\sim N\left(\left(\begin{array}{c}
                                   1 \\
                                   1 \\
                                   \vdots\\
                                   1 \\
                                 \end{array}
                               \right),\left(
                               \begin{array}{cccc}
                               1 & 0.4 & \cdots & 0.4 \\
                               0.4 & 1 & \cdots & 0.4 \\
                               \vdots & \vdots & \ddots & \vdots \\
                               0.4 & 0.4 & \cdots & 1 \\
                               \end{array}
                               \right)
                               \right)\\
\nonumber && p(y_1,y_2,\ldots,y_{10} |H_{0})\sim N\left(\left(\begin{array}{c}
                                   0 \\
                                   0 \\
                                   \vdots\\
                                   0 \\
                                 \end{array}
                               \right),\left(
                               \begin{array}{cccc}
                               0.6 & 0 & \cdots & 0 \\
                               0 & 0.6 & \cdots & 0 \\
                               \vdots & \vdots & \ddots & \vdots \\
                               0 & 0 & \cdots & 0.6 \\
                               \end{array}
                               \right)
                               \right)
\end{eqnarray}
\end{example}

%\begin{figure}[t]
%%\hrule %\vspace*{8cm}
%\vbox to 8cm{\vfill \hbox to \hsize{\hfill
%\scalebox{0.7}[0.7]{\includegraphics{10sensors_10n_optimal.eps}}
%\hfill}\vfill}
%%\hrule
%\caption{Ten-sensor ROC curves}\label{fig_02}
%\end{figure}
\begin{figure}[!htb]
  \centering
  \scalebox{1}[1]{
  \includegraphics[width=\hsize]{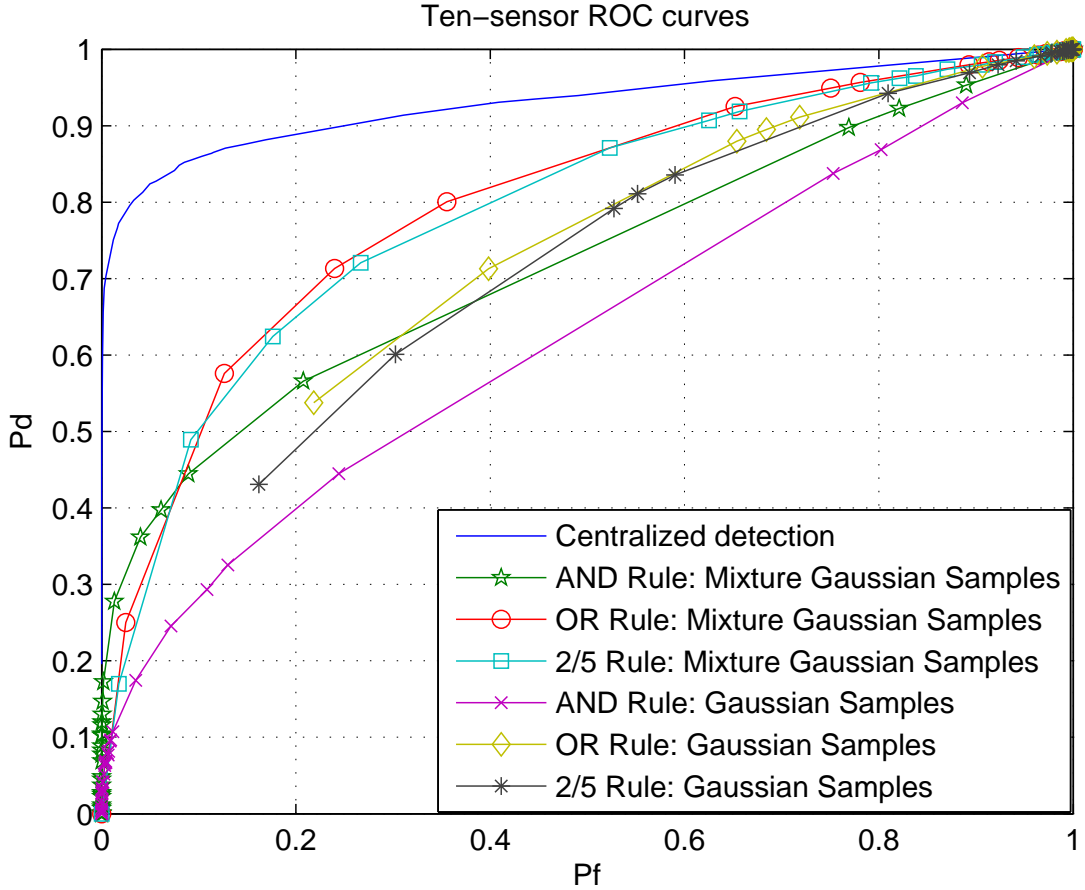}}
  \caption{Ten-sensor ROC curves.}
  \label{fig_02}
\end{figure}

In Figure \ref{fig_02}, the ROC curves for Centralized algorithm, Algorithm \ref{alg_1} with a mixture Gaussian trial distribution and Algorithm \ref{alg_1} with a Gaussian trial distribution are provided  where the AND, OR and 2 out of 5 (2/5) fusion rules are considered, respectively. For Algorithm \ref{alg_1}, we draw $N=1000$ samples from the trial distribution.
%, then draw $N_{c}=2000$ samples from $p(y_1,y_2,\ldots,y_{10} |H_{0})$, $p(y_1,y_2,\ldots,y_{10} |H_{1})$ for ROC curves respectively.
%For Riemann Gauss-Seidel iterative algorithm, we take a discretization step-size $\triangle_{1}=18/N=0.18$, $y_{i}\in[-8,10]$.
%In cost function (\ref{Eq_1_7}), let cost parameters $C_{00}=C_{11}=0$ and $C_{10}=C_{01}=1$.
%We use 29 pairs parameters $a1=1$, $b1=200000,10000,200,40,20,12,4.4,3.2,1.7,1,0.7,0.6,\\0.5,0.4,0.36,0.32,0.2,0.1$, respectively.
The initial values of the sensor rule are $I_{i}(y_{i})=I[3y_{i}-4]$, for $i=1,\cdots,L$.
%We denote the probability of a false alarm and the probability of detection as $P_f$ and $P_d$, respectively.

The solid line is the ROC curve calculated by the centralized algorithm. The star line, circle line and square line are the ROC curves for the fixed AND, OR and 2 out of 5 (2/5) fusion rule calculated by Algorithm \ref{alg_1} with Mixture-Gaussian trial distribution, respectively. The $\times$ line, diamond line and $*$ line line are the ROC curves for the fixed AND, OR and 2 out of 5 (2/5) fusion rule calculated by Algorithm \ref{alg_1} with Gaussian trial distribution, respectively.

From Figure \ref{fig_02}, we have the following observations:
\begin{itemize}
  \item The performance of Algorithm \ref{alg_1} with Mixture-Gaussian trial distribution is better than that of Algorithm \ref{alg_1} with Gaussian trial distribution. The reason may be that the optimal trial distribution in (\ref{Eq_1_17}) should be $g(y_1,y_2,\ldots,y_L)\varpropto |I_{\Omega_0}(y_1,y_2,\ldots,y_L)\hat{L}(y_1,y_2,\ldots,y_L)|$ (see, e.g., \cite{Liu01,Robert05}) and $|\hat{L}(y_1,y_2,\ldots,y_L)|=|ap(y_1,\ldots,y_L|H_1)-bp(y_1,\ldots,y_L|H_0)|$ which is similar to Mixture-Gaussian. Thus, the performance  based on Mixture-Gaussian trial distribution is better than that of Gaussian trial distribution.
  \item When probability of a false alarm $P_f$  is small, the performance of the fixed AND fusion rule is better than that of the fixed OR fusion rule and vice versa.
  \item For the same parameters, most of points of the AND fusion rule converge to the $(0, 0)$ and most of points of the OR fusion rule converge to the $(1, 1)$. The reason may be the AND fusion rule corresponds to a smaller  probability of a false alarm $P_f$ than that of the OR fusion rule.
\end{itemize}

\subsection{One hundred sensors}\label{sec_5_3}
\begin{example}Let us consider a surveillance model. There is a target/signal $s$ which may cross a surveillance region from one of 50 paths with an equal probability. 100 sensors are deployed on the 50 paths separately. Each path has two sensors. The 100 sensors transmit the decision 0 or 1 to the fusion center. We consider a given fusion rule that if there is only one path where two sensors make a decision (1, 1), then the fusion center makes a decision 1; otherwise, make a decision 0.

%As we all know, the hierarchical model can be used for large-scale sensor networks. Here we consider designing a multiple levels decision system. Several groups of sensors are spread in the space. Each group consists of multiple sensors can be regard as a distributed detection fusion system before sending its fused result to a final fusion center.
%
%The target can only be detected by one group of sensors with the same probability. We consider Monte Carlo importance sampling methods in one hundred sensor networks with inner AND outer OR fusion rule. The one hundred sensor networks set two sensors as a group defined as inner layer using AND fusion rule, then the fifty groups set as outer layer using OR fusion rule(i.~e.,the inner layer fusion result as input of outer layer).

The signal $s$ and observation noises $v_{1}$, $v_{2}$, \ldots, $v_{100}$
%\begin{eqnarray}
%\nonumber H_{1}:~~&&y_{i}=s+v_{i}, ~~ for ~ i=1,2 \\
%\nonumber         &&y_{i}=s+v_{i}, ~~ for ~ i=3,4 \\
%\nonumber         &&\cdots\cdots\cdots\\
%\nonumber         &&y_{i}=s+v_{i}, ~~ for ~ i=99,100\\
%\nonumber H_{0}:~~&&y_{i}=v_{i}, ~~ for ~ i=1,\ldots,100\qquad\qquad\qquad\qquad\qquad~~
%\end{eqnarray}
%where $s$, $v_{1}$, $v_{2}$ \ldots, $v_{100}$
are all mutually independent and
\begin{eqnarray}
\nonumber &&s\sim N(1,0.4), ~~ v_{i}\sim N(0,0.6), ~~ for ~ i=1,\ldots,100.
\end{eqnarray}
Thus, the two conditional probability density functions (pdfs) given $H_{0}$ and $H_{1}$ are
\begin{eqnarray}
\nonumber &&p(y_1,y_2,\ldots,y_{100} |H_{0})\sim N\left(\mu_0,~ \Sigma_0\right),\qquad\qquad\qquad\qquad~\\
\nonumber &&p(y_1,y_2,\ldots,y_{100} |H_{1})\sim \sum_{i=1}^{50}P\times N\left(\mu_{1,i},~\Sigma_{1,i}\right),
\end{eqnarray}
where $P=1/50$,
\begin{eqnarray}
\nonumber &&\mu_0=(0,\cdots,0)_{100\times1}',~~\quad~~\Sigma_0=diag(0.6,\cdots,0.6)_{100\times100},\\
\nonumber &&\mu_{1,1}=(\mu',0,\cdots,0)_{100\times1}', \cdots,\mu_{1,50}=(0,\cdots,0,\mu')_{100\times1}',\\
\nonumber &&\Sigma_{1,1}=diag(\Sigma,0.6,\cdots,0.6)_{100\times100},\cdots,\Sigma_{1,50}=diag(0.6,\cdots,0.6,\Sigma)_{100\times100},\\
\nonumber &&\mu=\left(\begin{array}{c}
                                   1 \\
                                   1 \\
                                 \end{array}
                               \right)_{2\times1},\Sigma=\left(
                               \begin{array}{cccc}
                               1 & 0.4 \\
                               0.4 & 1 \\
                               \end{array}
                               \right)_{2\times2}.
\end{eqnarray}
\end{example}

\begin{figure}[!htb]
  \centering
  \scalebox{1}[1]{
  \includegraphics[width=\hsize]{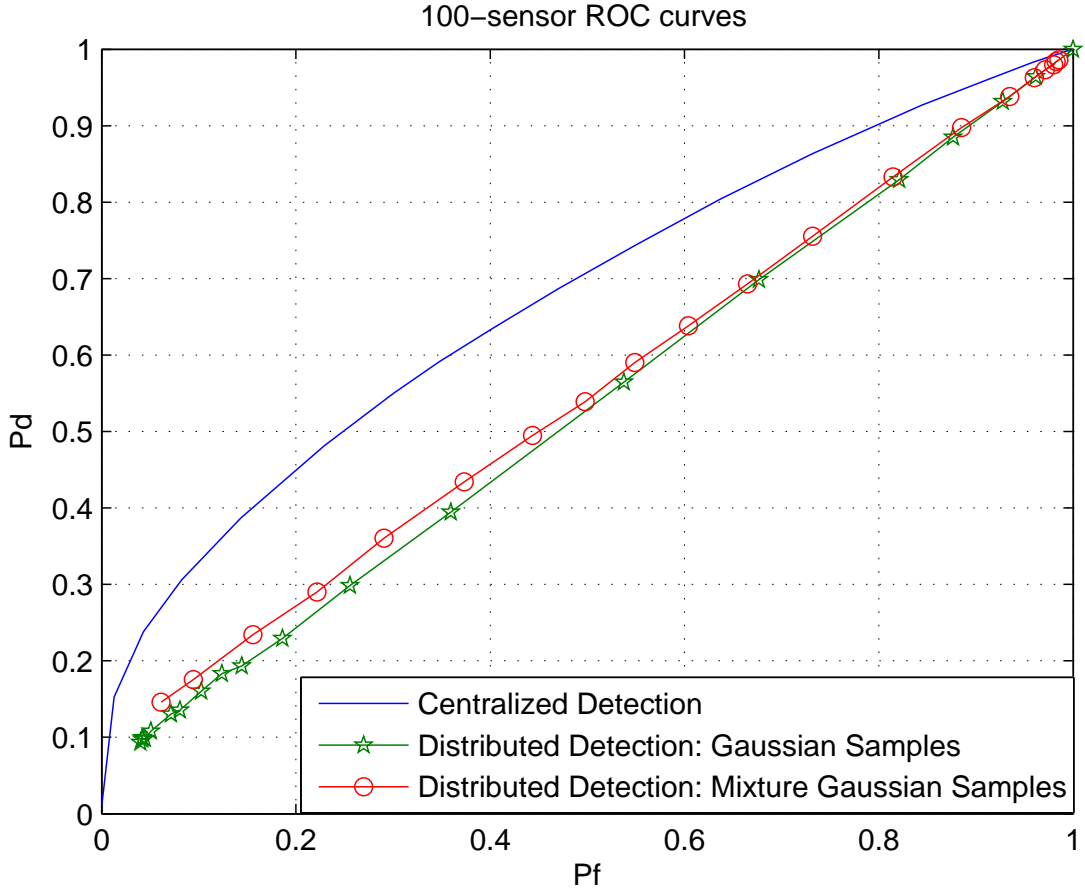}}
  \caption{One-hundred-sensor ROC curves}
  \label{fig_03}
\end{figure}

In Figure \ref{fig_03}, the ROC curves for Centralized algorithm, Algorithm \ref{alg_1} with a mixture Gaussian trial distribution and Algorithm \ref{alg_1} with a Gaussian trial distribution are provided. For Algorithm \ref{alg_1}, we draw $N=10000$ samples from the trial distribution to derive the optimal sensor decision rules.
% The AND fusion rule is used in each group. Their decision results are fused by OR fusion rule.
% This fusion rule is one of the \textcolor[rgb]{1.00,0.00,0.00}{${2^{2}}^{200}$} fusion rules for the parallel distributed fusion structure.
%, then draw $N_{c}=20000$ samples from $p(y_1,y_2,\ldots,y_{50} |H_{0})$, $p(y_1,y_2,\ldots,y_{50} |H_{1})$ for calculating ROC curves respectively.
%For Riemann Gauss-Seidel iterative algorithm, we take a discretization step-size $\triangle_{1}=18/N=0.18$, $y_{i}\in[-8,10]$.
%In cost function (\ref{Eq_1_7}), let cost parameters $C_{00}=C_{11}=0$ and $C_{10}=C_{01}=1$.
%We use 29 pairs parameters $a1=1$, $b1=200000,10000,200,40,20,12,4.4,3.2,1.7,1,0.7,0.6,\\0.5,0.4,0.36,0.32,0.2,0.1$, respectively.
%The initial values of the sensor rule are $I_{i}(y_{i})=I[3y_{i}-4]$, for $i=1,\cdots,L$.
%We denote the probability of a false alarm and the probability of detection as $P_f$ and $P_d$, respectively.

The solid line is the ROC curve calculated by the centralized algorithm. The circle line is the ROC curve for the fixed fusion rule by Algorithm \ref{alg_1} with Mixture-Gaussian trial distribution. The star line is the ROC curve for the fixed fusion rule calculated by with Gaussian trial distribution.

From Figure \ref{fig_03}, it can be seen that the performance of Algorithm \ref{alg_1} with Mixture-Gaussian trial distribution is better than that of Algorithm \ref{alg_1} with Gaussian trial distribution. The reason is similar to the case of two sensors or ten sensors.
This example also shows that the new method can be applied to large number of sensor networks when the fusion rule is fixed.

\section{Conclusion}\label{sec_6}
In the paper, we have proposed a Monte Carlo framework for the distributed detection fusion with high-dimension conditionally dependent observations. By using the Monte Carlo importance sampling, we derived a necessary condition for optimal sensor decision rules  so that a Gauss-Seidel optimization approach can be obtained to search the optimal sensor decision rules. We proved that the discretized algorithm is finitely convergent. The complexity of the new algorithm is order of $O(LN)$ compared with $O(LN^L)$ of the previous algorithm where $L$ is the number of sensors and $N$ is the sample size in the importance sampling draw. Thus, the proposed methods allows us to design the large sensor networks with general dependent observations. Furthermore, an interesting result is that, for the fixed AND or OR fusion rules, we have analytically derived the optimal solution in the sense of minimizing the approximated Bayesian cost function. In general, the solution of the Gauss-Seidel algorithm is only local optimal. However, in the new framework, we have proved that the solution of Gauss-Seidel algorithm is same as the analytically optimal solution in the case of the AND or OR fusion rule. The typical examples with dependent observations and large number of sensors are examined under this new framework. The results of numerical examples demonstrate the effectiveness of the new algorithm class. Future work will involve the generalization of the Monte Carlo framework for parallel networks to all kinds of networks. When the number of the sensors is very large, how to optimize the fusion rule is another challenging problem.

\end{document}